\documentclass[runningheads]{llncs}

\usepackage{graphicx} 
\graphicspath{ {./images/} }
\usepackage[rightcaption]{sidecap}
\usepackage{wrapfig}
\usepackage{amsmath}
\newcommand{\E}{{\bf E}}
\newcommand{\ignore}[1]{}

\begin{document}
\title{On Detecting Some Defective Items in\\ Group Testing}
%
%
\author{Nader H. Bshouty\orcidID{0009-0007-7356-7824} \and \\
Catherine A. Haddad-Zaknoon\orcidID{0009-0008-1503-594X}}
\authorrunning{N. H. Bshouty and C. A. Haddad-Zaknoon}
%
\institute{Technion - Israel Institute of Technology, Haifa, Israel
\email{\{bshouty,catherine\}@cs.technion.ac.il}}
\maketitle              
\begin{abstract}

Group testing is an approach aimed at identifying up to $d$ defective items among a total of $n$ elements. This is accomplished by examining subsets to determine if at least one defective item is present. In our study, we focus on the problem of identifying a subset of $\ell\leq d$ defective items. We develop upper and lower bounds on the number of tests required to detect $\ell$ defective items in both the adaptive and non-adaptive settings while considering scenarios where no prior knowledge of $d$ is available, and situations where an estimate of $d$ or at least some non-trivial upper bound on $d$ is available. 

\ignore{When no prior knowledge on $d$ is available, we prove a lower bound of $
\Omega(\frac{\ell \log^2n}{\log \ell +\log\log n})$ tests in the randomized non-adaptive settings and an upper bound of $O(\ell \log^2 n)$ for the same settings. Moreover, we show that any non-adaptive deterministic algorithm must ask $\Theta(n)$ tests. For adaptive algorithms, we prove a tight bound of $\Theta(\ell\log{(n/\ell)})$ for the deterministic case, and $\Theta(\ell\log{(n/d)})$ for the randomized settings.}

When no prior knowledge on $d$ is available, we prove a lower bound of $
\Omega(\frac{\ell \log^2n}{\log \ell +\log\log n})$ tests in the randomized non-adaptive settings and an upper bound of $O(\ell \log^2 n)$ for the same settings. Furthermore, we demonstrate that any non-adaptive deterministic algorithm must ask $\Theta(n)$ tests, signifying a fundamental limitation in this scenario. For adaptive algorithms, we establish tight bounds in different scenarios. In the deterministic case, we prove a tight bound of $\Theta(\ell\log{(n/\ell)})$. Moreover,  in the randomized settings, we derive a tight bound of $\Theta(\ell\log{(n/d)})$.

When $d$, or at least some non-trivial estimate of $d$, is known, we prove a tight bound of $\Theta(d\log (n/d))$ for the deterministic non-adaptive settings, and $\Theta(\ell\log(n/d))$ for the randomized non-adaptive settings. In the adaptive case, we present an upper bound of $O(\ell \log (n/\ell))$ for the deterministic settings, and a lower bound of $\Omega(\ell\log(n/d)+\log n)$. Additionally, we establish a tight bound of $\Theta(\ell \log(n/d))$ for the randomized adaptive settings.

\end{abstract}

\section{Introduction}
Group testing is a technique for identifying a subset of items known as \emph{defective items set} within a large amount of items using small number of tests called \emph{group tests}. A group test is a subset of items, where the test result is \emph{positive} if the subset contains at least one defective item and \emph{negative} otherwise. Formally, let $X=\{1,2,\ldots,n\}$ be a set of items, and  $I \subseteq X$ is the set of defectives. A group test is set $Q \subseteq X$. The answer of the test $Q$ with respect to the defective set $I$ is $1$ if $Q \cap I \not=\emptyset$, and $0$ otherwise.
Throughout the paper, we denote the number of defective items by $d$ and the number of items by $n:=|X|$.

Group testing was formerly purposed by Robert Dorfman \cite{D43}, for economizing mass blood testing during WWII. Since it was initially proposed, group testing methods have been utilized in a variety of applications including DNA library screening, product quality control and neural group testing for accelerating deep learning \cite{DH00,DH06,SG59,KS64,JW85,Li62,HL02,CM05,LiJZ21}. Among its recent applications, group testing has been advocated for accelerating mass testing for COVID-19 PCR-based tests around the world ~\cite{EHWB20,CSPL20,BENAMI2020,GC20,MCR20,SHGOY20,YK20,CH22}.

Several settings for group testing has been developed over the years. 
The distinction between \emph{ adaptive} and \emph{non-adaptive} algorithms is widely considered. In {\it adaptive} algorithms, the tests can depend on the answers to the previous ones. In the non-adaptive algorithms, they are independent of the previous one and; therefore, one can make all the tests in one parallel step. An $r-$round algorithm is an intermediate approach. We say that an adaptive algorithm is an \emph{$r$-round } if it runs in $r$ stages where each stage is non-adaptive. That is, the queries may depend on the answers of the queries in previous stages, but are independent of the answers of the current stage queries. 

Unlike conventional group testing, we consider the problem of finding only a subset of size $\ell\leq d$ from the $d$ defective items. In~\cite{ACL12}, the authors solve this problem for $\ell=1$ in the adaptive deterministic settings. They prove a tight bound of $\log{n}$ tests. For general $\ell$, they prove an upper bound of $O(\ell\log{n})$. Both results are derived under the assumption that $d$ is known exactly to the algorithm. When no prior knowledge on $d$ is available, Katona,~\cite{KATONA}, proves that for the deterministic non-adaptive settings, finding a single defective item (i.e. $\ell=1$) requires $n$ tests. For $\ell=1$ in the adaptive deterministic settings, however, Katona proves a tight bound of $\Theta(\log n)$ tests, and for deterministic $2-$round settings (and $\ell=1$) a tight bound of $\Theta\left (\sqrt{n}\right )$ is given. Gerbner and Vizer, ~\cite{GV16}, generalized Katona's result for the deterministic $r-$round settings for all $\ell\geq 1$. They give a lower bound of $\Omega\left (r(\ell n)^{1/r}-r\ell\right )$ and an upper bound of $O\left (r{(\ell^{r-1}n)^{1/r}}\right)$.

In this paper, we  study the test complexity of this problem in adaptive, non-adaptive, randomized and deterministic settings. We establish lower and upper bounds for the scenario where no prior knowledge of the value of $d$ is available. Moreover, we study the same problem when there is either an estimation or at least an upper bound on the value of $d$. In the literature, some results assume that $d$ is known {\it exactly}, or some upper bound on the number of defective items is known in advance to the algorithm. In practice, only an estimate of $d$ is known. In this paper, when we say that $d$ is known in advance to the algorithm, we assume that some estimate $D$ that satisfies $d/4\le D\le 4d$ is known to the algorithm. We will also assume that $\ell\le d/4$. Otherwise, use the algorithm that detects all the defective items. Our results are summarized in the following subsection.
\subsection{Detecting $\ell$ Defective Items from $d$ Defective Items}

The results are in the Table in Figure~\ref{Table2}. The results marked with $\star$ are the most challenging results of this paper.
\begin{figure}[h]
\centering
{\includegraphics[trim={0 1cm 0 0}, width=0.9\textwidth]{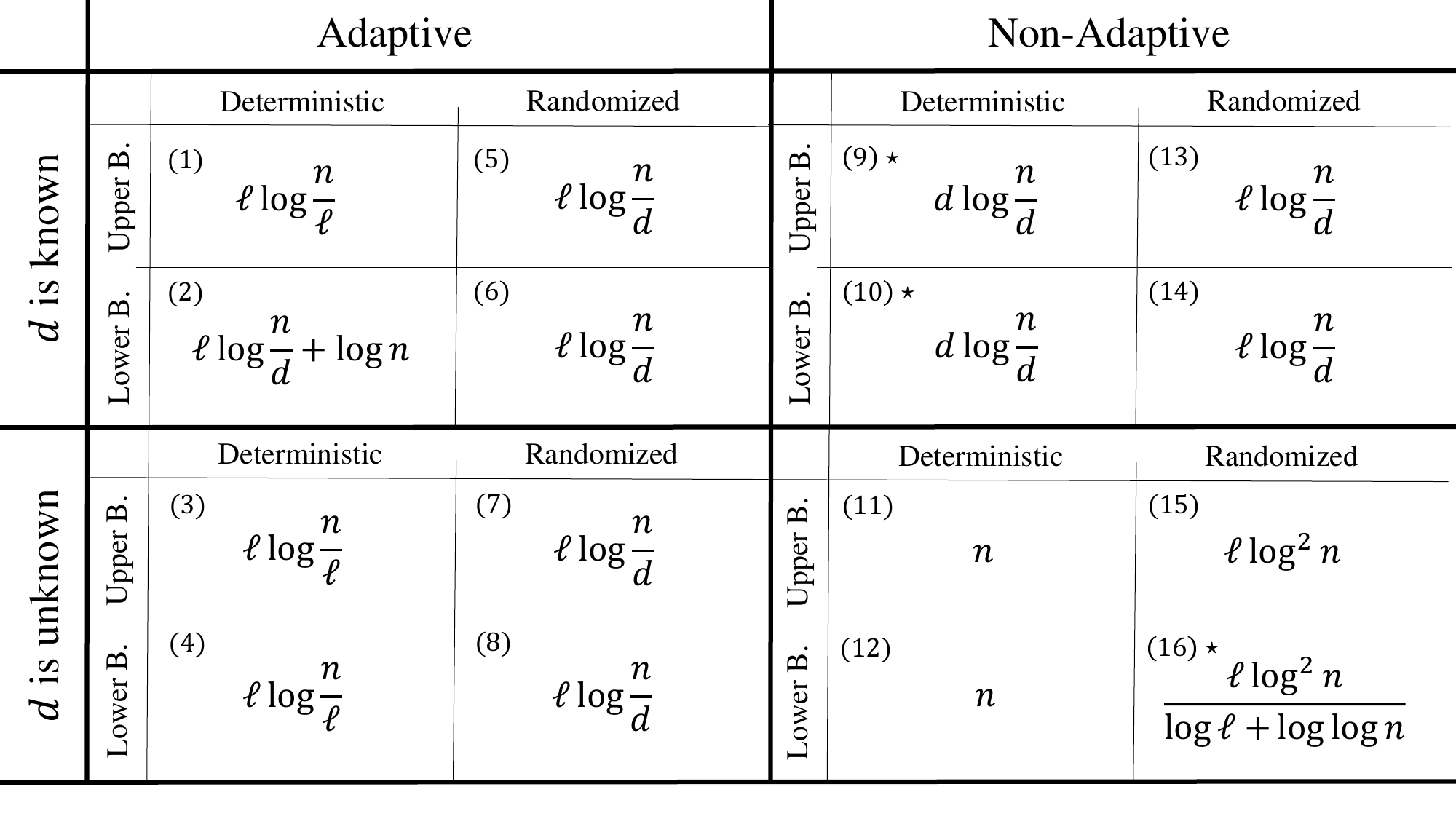}}
\caption{Results of detecting $\ell$ from $d$ defective items. In (1) and (2), the bounds are asymptotically tight when there is a constant $c<1$, such that $d\le n^{c}$. The results marked with $\star$ are the most difficult and challenging results of this paper.}\label{Table2}
\end{figure}
\begin{itemize}
    \item In (1) and (2), (in the table in Figure~\ref{Table2}) the algorithm is deterministic adaptive, and $d$ is known in advance to the algorithm. The upper bound is $\ell\log(n/\ell)+O(\ell)$. The algorithm  splits the items into equal sizes $\ell$ disjoint sets and uses binary search to detect $\ell$ defective items in the sets. 

    The lower bound is $\max(\ell\log(n/d),\log n)$. In what follows, when we say input, we mean $I\subset [n]$, $|I|=d$ (the defective items)\footnote{A lower bound for the number of tests when the algorithm knows exactly $d$, is also a lower bound when the algorithm knows some estimate of $d$ or does know $d$.} and, when we say output, we mean a subset $L\subset I$ of size $\ell$. A set of $\ell$ items can be an output of at most ${n-\ell\choose d}$ inputs. This gives a lower bound for the number of outputs of the algorithm, which, in turn, (its $\log$) gives a lower bound $\ell\log(n/d)-\ell$ for the test complexity.  For the lower bound $\log n$, we show that if the number of possible outputs of the algorithm is less than $n-d$, then one can construct a size $d$ input $I$ that contains no output. Therefore, the test complexity is at least $\log(n-d)$. 

    \item In (3) and (4), the algorithm is deterministic adaptive, and $d$ is unknown to the algorithm. The upper bound is $\ell\log(n/\ell)+O(\ell)$. The algorithm in (1) also works when $d$ is unknown to the algorithm. The lower bound follows from (2) when we choose $d=4\ell$.
    \item In (5) and (6), the algorithm is randomized adaptive, and $d$ is known in advance to the algorithm.  The upper bound is $\ell\log(n/d)+O(\ell)$. The algorithm uniformly at random chooses each element in $X$ with probability $O(\ell/d)$ and puts the items in $X'$.  We show that, with high probability, $X'$ contains $\ell$ defective items and $|X'|=O(n\ell/d)$. Then, the algorithm deterministically detects $\ell$ defective items in $X'$ using the algorithm of (3). This gives the result.
    
     For the lower bound, $\ell\log(n/d)-1$, we use the same argument as in the proof of (2) in Figure~\ref{Table2} with Yao's principle. 
    \item In (7) and (8), the algorithm is randomized adaptive, and $d$ is unknown to the algorithm. The upper bound is $\ell\log(n/\ell)+O(\ell+\log\log(\min(n/d,d)))=O(\ell\log(n/\ell))$. We first give a new algorithm that estimates $d$ that uses $\log\log(\min(n/d,d))$ tests and then use the algorithm in (5). The lower bound follows from (4).
    \item In (9) and (10), the algorithm is deterministic non-adaptive, and $d$ is known in advance to the algorithm.

    For the upper bound, we first define the $(2d,d+\ell)$-restricted weight one $t\times n$-matrix. This is a $t\times n$ 0-1-matrix such that any set of $2d$ columns contains at least $d+\ell$ distinct weight one vectors. Using this matrix, we show how to detect $\ell$ defective items with $t$ tests. Then we show that there is such a matrix with  $t=O(d\log (n/d))$ rows. 

    We then give a tight lower bound. We show that such construction is not avoidable. From any non-adaptive algorithm, one can construct a $(2d,1)$-restricted weight one $t\times n$-matrix. We then show that such a matrix must have at least $t=\Omega(d\log (n/d))$ rows. 
    \item In (11) and (12), the algorithm is deterministic non-adaptive, and $d$ is unknown to the algorithm.

    We show that any such algorithm must test all the items individually. 
    \item In (13) and (14), the algorithm is randomized non-adaptive, and $d$ is known in advance to the algorithm. 

    The upper bound is $O(\ell\log(n/d))$. The algorithm runs $t=O(\ell)$ parallel iterations. At each iteration, it uniformly at random chooses each element in $X=[n]$ with probability $O(1/D)$ and puts it in $X'$. With constant probability, $X'$ contains one defective item. Then it uses the algorithm that tests if $X'$ contains exactly one defective item and detects it. 
    
    The lower bound follows from (6). 
    \item In (15) and (16), the algorithm is randomized non-adaptive, and $d$ is unknown to the algorithm.

    The upper bound is $O(\ell\log^2 n)$. The algorithm runs the non-adaptive algorithm that gives an estimation $d/2\le D\le 2d$ of $d$  \cite{Bshouty19,DamaschkeM12,FalahatgarJOPS16}, and, in parallel, it runs $\log n$ randomized non-adaptive algorithms that find $\ell$ defective items  assuming $D=2^i$ for all $i=1,2,\ldots,\log n$ (the algorithm in (13)). 

    The lower bound is $\tilde \Omega(\ell\log^2n)$. The idea of the proof is the following. Suppose a randomized non-adaptive algorithm exists that makes $\ell\log^2n/(c\log R)$ tests where $R=\ell\log n$ and $c$ is a large constant. We partition the internal $[0,n]$ of all the 
    possible sizes of the tests $|Q|$ into $r=\Theta(\log n/\log R)$ disjoint sets $N_i=\{m|n/R^{8i+8}<m\le n/R^{8i}\}$, $i\in [r]$. We then show that, with high probability, there is an interval $N_j$ where the algorithm makes at most $(\ell/c)\log n$ tests $Q$ that satisfy $|Q|\in N_j$. We then show that if we choose uniformly at random a set of defective items $I$ of size $d=(\ell\log n)^{8j+4}$, then with high probability, all the tests $Q$ of size $|Q|<n/(\ell\log n)^{8j+8}$ give answers $0$, and all the tests $Q$ of size $|Q|>n/(\ell\log n)^{8j}$ give answers $1$. So, the only useful tests are those in $N_j$, which, by the lower bound in (14) (and some manipulation), are insufficient for detecting $\ell$ defective items. 
\end{itemize}

\subsection{Known Results for Detecting all the Defective Items}
The following results are known for detecting all the $d$ defective items. See the Table in Figure~\ref{Table1}.

\begin{figure}[h]
\centering
{\includegraphics[trim={0 1cm 0 0}, width=0.9\textwidth]{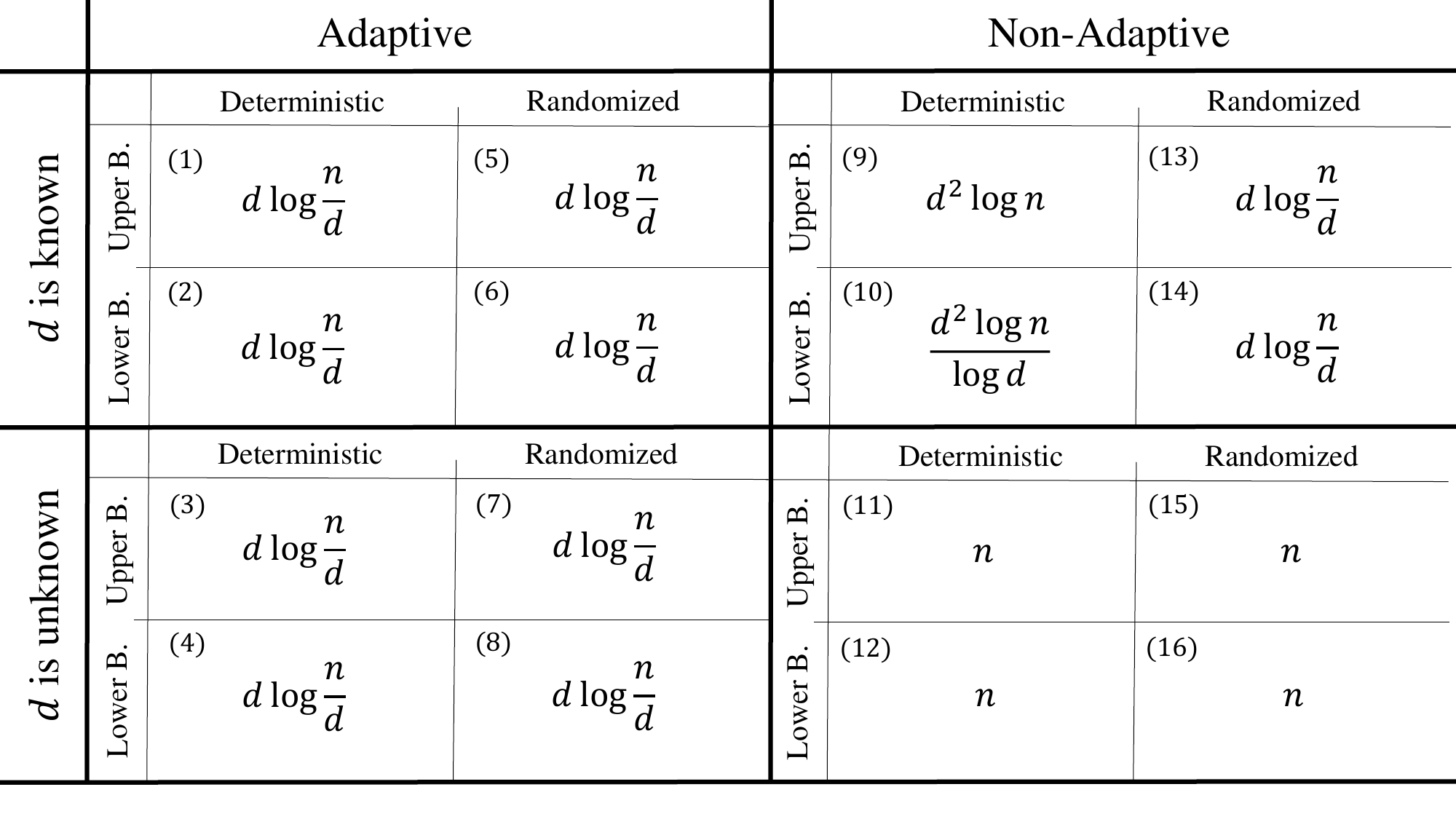}}
\caption{Results for the test complexity of detecting the $d$ defective items. The lower bounds are in the $\Omega$-symbol and the upper bound are in the $O$-symbol}\label{Table1}
\end{figure}
\begin{itemize}
\item In (1) and (2) (in the table in Figure~\ref{Table1}), the algorithm is deterministic adaptive, and $d$ is known in advance to the algorithm. 
The best lower bound is the information-theoretic lower bound $\log {n\choose d}\ge d\log(n/d)+\Omega(d)$. Hwang in~\cite{Hwang72} gives a generalized binary splitting algorithm that makes $\log {n\choose d}+d-1=d\log(n/d)+O(d)$ tests. 
\item In (3) and (4), the algorithm is deterministic adaptive, and $d$ is unknown to the algorithm.  The upper bound $d\log(n/d)+O(d)$ follows from~\cite{Bar-NoyHKK94,ChengDX14,DuH93,DuP94,DuXSC94,SchlaghoffT05,WuCD22} and the best constant currently known in $O(d)$ is $5-\log 5\approx 2.678$~\cite{WuCD22}. The lower bound follows from (2). In \cite{BshoutyBHHKS18}, Bshouty et al. show that estimating the number of defective items within a constant factor requires at least $\Omega(d\log(n/d))$ tests.
\item 
In (5) and (6), the algorithm is randomized adaptive, and $d$ is known in advance. The upper bound follows from (1). The lower bound follows from Yao's principle with the information-theoretic lower bound. 
\item In (7) and (8), the algorithm is randomized adaptive, and $d$ is unknown to the algorithm. The upper bound $d\log(n/d)+O(d)$ follows from (3). The lower bound follows from (6). 
\item In (9) and (10), the algorithm is deterministic non-adaptive, and $d$ is known in advance to the algorithm. The lower bound  $\Omega(d^2\log n/\log d)$ is proved in~\cite{ChenH07a,DyachkovR82,Furedi96a,Ruszinko94}. A polynomial time algorithm that constructs  an algorithm that makes $O(d^2\log n)$ tests was first given by Porat and Rothschild~\cite{PoratR11}.
\item In (11) and (12), the algorithm is deterministic non-adaptive and $d$ is unknown to the algorithm. In \cite{Bshouty19}, Bshouty shows that estimating the number of defective items within a constant factor requires at least $\Omega(n)$ tests. The upper bound is the trivial bound of testing all the items individually.
\item In (13) and (14), the algorithm is randomized non-adaptive, and $d$ is known in advance to the algorithm. The lower bound follows from (6). The upper bound is $O(d\log (n/d))$. The constant in the O-symbol was studied in~\cite{BaldingBTK96,BshoutyDKS17,BshoutyHH20,DuH06} and referenced within. The best constant known in the $O$-symbol is $\log e\approx 1.443$~\cite{BshoutyHH20}.

\item In (15) and (16), the algorithm is randomized non-adaptive, and $d$ is unknown to the algorithm. The lower bound $\Omega(n)$ follows from Yao's principle and the fact that, for a random uniform $i\in[n]$, to detect the defective items $[n]\backslash \{i\}$, we need at least $\Omega(n)$ tests. The upper bound is the trivial bound of testing all the items individually.
\end{itemize}

\subsection{Applications}

In many cases, the detection of a specific number of defective items, $\ell$,  is of utmost importance due to system limitations or operational requirements. For instance, in scenarios like blood tests or medical facilities with limited resources such as ventilators, doctors, beds, or medicine supply, it becomes crucial to employ algorithms that can precisely identify $\ell$ defectives instead of detecting all potential cases. This targeted approach offers significant advantages in terms of efficiency, as the time required to detect only $\ell$ defective items is generally much shorter than the time needed to identify all defects. By focusing on any subset of $\ell$ defectives, the algorithms proposed in this paper offer more efficient procedures.

In the following, we present some real-world applications that demonstrate the practical use of the problem of finding only a subset of size $\ell\leq d$ of the defective set items.
\subsubsection{Identifying a subset of samples that exhibit a PCR-detectable syndrome}

 \emph{Polymerase Chain Reaction} or \emph{PCR}  testing is a widely used laboratory technique in molecular biology. This technique is used to amplify specific segments of DNA or RNA in a sample, and therefore, allowing for detection, quantification and analyses of these specific genetic sequences~\cite{YK20,DuH06,CH22}. PCR tests can be designed to identify various organisms, including pathogens such as viruses or bacteria (e.g. COVID-19), by targeting their unique DNA or RNA signatures.
 Although PCR tests are associated with high costs and time consumption, they are extensively utilized in a wide range of fields, including medical diagnostics, research laboratories, forensic analysis, and other applications that demand accurate and sensitive detection of genetic material. This popularity is primarily attributed to their exceptional accuracy. To enhance the efficiency and cost-effectiveness of PCR testing, group testing methodologies can be applied to PCR testing. Applying group testing to PCR involves combining multiple samples into a single test sample. The combined sample, also called the \emph{group test}, is then examined. If the sample screening indicates an infectious sample, this implies that at least one of the original samples is infected. Conversely, if none of the samples in the combined sample exhibit signs of infection, then none of the individual samples are infected. Typically, PCR tests are conducted by specialized machines capable of simultaneously performing approximately $96$ tests. Each test-run  can span over several hours. Therefore, when applying group testing to accelerate PCR process, it is recommended to employ non-adaptive methodologies. 
 
Assume that a scientific experiment need to be conducted over a group of study participants to examine the efficiency of a new drug developed for medicating some disease related to bacterial or virus infection. Suppose that a PCR  test is required to check whether a participant is affected by the disease or not. Moreover, assume that the number of the  participants that volunteered for the experiment is $n$ and the incidence rate of the infection among them is known in advance, denote that by $p$. Therefore, an approximation of the number of infected participants can be derived from $n$ and $p$, denote that value by $d$. In situations where logistic constraints necessitate selecting a limited number of infected individuals, specifically $\ell\leq d$, to participate in an experiment, a non-adaptive group testing algorithm for identifying $\ell$ defectives (virus carriers) from $n$ samples when $d$ is known can be employed.

\subsubsection{Abnormal event detection in surveillance camera videos}
Efficiently detecting abnormal behavior in surveillance camera videos plays a vital role in combating  crimes. These videos are comprised of a sequence of continuous images, often referred to as \emph{frames}.  The task of identifying suspicious behavior within a video is equivalent to searching for abnormal behavior within a collection of frames. Training \emph{deep  neural networks} (shortly, DNN) for automating  suspicious image recognition is currently a widely adopted approach for the task ~\cite{XIE17,wk19,KUP22}. By utilizing the trained DNN, it becomes possible to classify a new image and determine whether it exhibits suspicious characteristics or not. However, once the training process is complete, there are further challenges to address, specially when dealing with substantial amount of images that need to be classified via the trained network. In this context, \emph{inference} is the process of utilizing the trained model to make predictions on new data that was not part of the training phase.  Due to the complexity of the DNN, inference time of images  can cost hundreds of seconds of GPU time for a single image. Long inference time poses challenges in scenarios where real-time or near-real-time processing is required, prompting the need for optimizing and accelerating the inference process. 

The detection of abnormal behavior in surveillance camera videos is often characterized by an imbalanced distribution of frames portraying abnormal behavior, also called \emph{abnormal frames},  in relation to the total number of frames within the video.
Denote the total number of frames in a video by $n$ and the number of abnormal frames by $d$. To identify suspicious behavior in a video, the goal is to find at least one abnormal frame among the $d$ frames. In most cases, we cannot assume any  non-trivial upper bound or estimation of any kind for $d$. Therefore, applying non-adaptive group testing algorithms for finding $\ell<d$ defectives when $d$ is unknown best suits this task.

It is unclear, however, how group testing can be applied to instances like images. Liang and Zou,~\cite{LiJZ21},  proposed three different methods for pooling image instances: 1) merging samples in the pixel space, 2) merging samples in the feature space, and 3) merging samples hierarchically and recursively at different levels of the network. For each grouping method, they provide network enhancements    that ensure that the group testing paradigm continues to hold. This means that a  positive prediction is inferred on a group if and only if it contains at least one positive image (abnormal frame). 
 \section{Definitions and Preliminary Results}
In this section, we give some definitions and preliminary results that we will need in the rest of the paper.

Let $X=[n]:=\{1,2,\ldots,n\}$ be a set of items that contains defective items $I \subseteq X$. In Group testing, we test a subset $Q \subseteq X$ of items, and the answer to the test is $T_I(Q):=1$ if $Q \cap I \not=\emptyset$, and $T_I(Q)=0$ otherwise. 

We will use the following version of Chernoff's bound.
\begin{lemma}\label{Chernoff}{\bf Chernoff's Bound}. Let $X_1,\ldots, X_m$ be independent random variables taking values in $\{0, 1\}$. Let $X=\sum_{i=1}^mX_i$ denotes their sum, and $\mu = \E[X]$ denotes the sum's expected value. Then
\begin{equation}\label{Chernoff1}
\Pr[X>(1+\lambda)\mu]\le \left(\frac{e^{\lambda}}{(1+\lambda)^{(1+\lambda)}}\right)^{\mu}\le e^{-\frac{\lambda^2\mu}{2+\lambda}}\le 
\begin{cases}
e^{-\frac{\lambda^2\mu} {3}} & \mbox{if \ }  0< \lambda\le 1 
\\
e^{-\frac{\lambda \mu}{3}} & \mbox{if \ } \lambda>1.
\end{cases}
\end{equation}
In particular,
\begin{eqnarray}\label{Chernoff2}\Pr[X>\Lambda]\le \left(\frac{e\mu}{\Lambda}\right)^{\Lambda}.\end{eqnarray}

For $0\le \lambda\le 1$ we have
\begin{eqnarray}\label{Chernoff3}
\Pr[X<(1-\lambda)\mu]\le \left(\frac{e^{-\lambda}}{(1-\lambda)^{(1-\lambda)}}\right)^{\mu}\le e^{-\frac{\lambda^2\mu}{2}}.
\end{eqnarray}
\end{lemma}

The following lemma follows from \cite{BshoutyBHHKS18,FalahatgarJOPS16}. 
\begin{lemma}\label{Estd0}
Let $\epsilon<1$ be any positive constant. There is a polynomial time adaptive algorithm that makes $O(\log\log d+\log(1/\delta))$ expected number of tests and with probability at least $1-\delta$ outputs $D$ such that $(1-\epsilon)d\le D\le (1+\epsilon)d$.
\end{lemma}
In Appendix~\ref{A}, we use a similar technique to prove:
\begin{lemma}\label{Estd}
Let $\epsilon<1$ be any positive constant. There is a polynomial time adaptive algorithm that makes $O(\log\log (\min(d,n/d))+\log(1/\delta))$ expected number of tests and with probability at least $1-\delta$ outputs $D$ such that $(1-\epsilon)d\le D\le (1+\epsilon)d$.
\end{lemma}

In \cite{Bshouty19,DamaschkeM12,FalahatgarJOPS16}, the following is proved
\begin{lemma}\label{NonAFd}
There is a polynomial time non-adaptive randomized algorithm that makes $O(\log(1/\delta)\log n)$ tests and, with probability at least $1-\delta$, finds an integer $D$ that satisfies $d/2<D<2d$.
\end{lemma}

\section{Adaptive and Deterministic}
In this section, we study the test complexity of adaptive deterministic algorithms.

We first prove the following upper bound. This proves results (1) and (3) in Figure~\ref{Table2}. Here $d$ can be known or unknown to the algorithm. 
\begin{theorem}\label{ADKU}
Let $d\ge \ell$. There is a polynomial time adaptive deterministic algorithm that detects $\ell$ defective items and makes at most $\ell\log(n/\ell)+3\ell=O(\ell\log (n/\ell))$ tests.
\end{theorem}
\begin{proof}
We first split the items $X=[n]$ to $\ell$ disjoint sets $X_1,\ldots,X_\ell$ of (almost) equal sizes (each of size $\lfloor n/\ell\rfloor$ or $\lceil n/\ell\rceil$). Then we use the binary search algorithm (binary splitting algorithm) for each $i$ to detect all the defective items in $X_i$ until we get $\ell$ defective items.

Each binary search takes at most $\lceil\log (n/\ell)\rceil+1$ tests, and testing all $X_i$ takes at most $\ell$ tests. \qed
\end{proof}

We now prove the lower bound. This proves (2) in Figure~\ref{Table2}. We remind the reader that when we say that $d$ is known in advance to the algorithm, we mean that an estimate $D$ that satisfies $d/4\le D\le 4d$ is known to the algorithm. The following lower bound holds even if the algorithm knows $d$ exactly in advance. 
\begin{theorem} \label{ADKL}
Let $\ell\le d\le n/2$ and $d$ be known in advance to the algorithm. Any adaptive deterministic algorithm that detects $\ell$ defective items must make at least $\max(\ell\log(n/d), \log n-1)=\Omega(\ell\log(n/d)+\log n)$ tests.
\end{theorem}
\begin{proof}
Let $A$ be an adaptive deterministic algorithm that detects $\ell$ defective items. Let $L_1,\ldots,L_t$ be all the possible $\ell$-subsets of $X$ that $A$ outputs. Since the algorithm is deterministic, the test complexity of $A$ is at least $\log t$. Since $L_i\subseteq I$ (the set of $d$ defective items), each $L_i$ can be an output of at most ${n-\ell\choose d-\ell}$ sets $I$. Since the number of possible sets of defective items $I$ is ${n\choose d}$, we have
$$t\ge \frac{{n\choose d}}{{n-\ell\choose d-\ell}}\ge \frac{n(n-1)\cdots(n-\ell+1)}{d(d-1)\cdots (d-\ell+1)}\ge \left(\frac{n}{d}\right)^\ell.$$ Therefore the test complexity of $A$ is at least $\log t\ge \ell \log(n/d).$

We now show that $t> n-d$. Now suppose, to the contrary, that $t\le n-d$. Choose any $x_i\in L_i$ and consider any $S\subseteq X\backslash \{x_i|i\in [t]\}$ of size $d$. For the set of defective items $I=S$, the algorithm outputs some $L_i$, $i\in [t]$. Since $L_i\not\subseteq S$, we get a contradiction. Therefore, $t> n-d$ and $\log t> \log(n-d)\ge \log(n/2)=\log n-1$.\qed
\end{proof}
Note that the upper bound $O(\ell\log (n/\ell))$ in Theorem~\ref{ADKU} asymptotically matches the lower bound $\Omega(\ell\log(n/d))$ in Theorem~\ref{ADKL} when $d=n^{o(1)}$. 

The following Theorem proves result~(4) in Figure~\ref{Table1}.
\begin{theorem}
Let $\ell\le d\le n/2$ and $d$ be unknown to the algorithm. Any adaptive deterministic algorithm that detects $\ell$ defective items must make at least $\ell\log(n/\ell)$ tests.
\end{theorem}
\begin{proof}
Since the algorithm works for any $d$, we let $d=4\ell$. Then by the first bound in Theorem~\ref{ADKL}, the result follows. \qed
\end{proof}
\section{Adaptive and Randomized}
In this section, we study the test complexity of adaptive randomized algorithms. 

The following theorem proves the upper bound when $d$ is known in advance to the algorithm. This proves result (5) in Figure~\ref{Table2}.
\begin{theorem}\label{ARKU}
Let $\ell\le d/2$. Suppose some integer $D$ is known in advance to the algorithm where $d/4\le D\le 4d$. There is a polynomial time adaptive randomized algorithm that makes $\ell\log(n/d)+\ell\log\log(1/\delta)+O(\ell)$ tests and, with probability at least $1-\delta$, detects~$\ell$ defective items.
\end{theorem}
\begin{proof} Let $c=32\log(2/\delta)$. If $D<c\ell$, we can use the deterministic algorithm in Theorem~\ref{ADKU}. The test complexity is $\ell\log(n/\ell)+2\ell=\ell\log(cn/D)+2\ell=\ell\log(n/d)+\ell\log\log(1/\delta)+O(\ell)$.

If $D>c\ell$, then the algorithm uniformly at random chooses each element in $X$ with probability $c\ell/D<1$ and puts the items in $X'$. If $|X'|\le 3c\ell n/D$, then deterministically detects $\ell$ defective items in $X'$ using Theorem~\ref{ADKU}.

Let $Y_i$ be an indicator random variable that is $1$ if the $i$th defective item is in $X'$ and $0$ otherwise. Then $\E[Y_i]=c\ell/D$. 
The number of defective items in $X'$ is $Y=Y_1+\cdots+Y_d$ and $\mu:=\E[Y]=cd\ell/D\ge c\ell/4$. 
By Chernoff's bound, we have $\Pr[Y<\ell]\le e^{-(1-4/c)^2c\ell/8}<e^{-c\ell/32}\le\delta/2$. Also, $\E[|X'|]=c\ell n/D$, and by Chernoff's bound, $\Pr[|X'|>3c\ell n/D]\le (e/3)^{3c\ell n/D}\le \delta/2.$
Therefore, with probability at least $1-\delta$, the number of defective items in $X'$ is at least $\ell$ and $|X'|\le 3c\ell n/D$. Therefore, with probability at least $1-\delta$, the algorithm detects $\ell$ defective items. 

Since $|X'|\le 3c\ell n/D\le 12c\ell n/d$, by Theorem~\ref{ADKU}, the test complexity is at most 
$\ell\log(|X'|/\ell)+2\ell=\ell\log(n/d)+\ell\log\log(1/\delta)+O(\ell).$\qed
\end{proof}

We now prove the lower bound when $d$ is known in advance to the algorithm. This proves results (6) and (8) in Figure~\ref{Table2}.
\begin{theorem} \label{ARKL}
Let $\ell\le d\le n/2$ and $d$ be known in advance to the algorithm. Any adaptive randomized algorithm that, with probability at least $2/3$, detects $\ell$ defective items must make at least $\ell\log(n/d)-1$ tests.
\end{theorem}
\begin{proof}
We use Yao's principle in the standard way. Let $A(s,I)$ be any adaptive randomized algorithm that, with probability at least $2/3$, detects $\ell$ defective items. Here $s$ is the random seeds, and $I$ is the set of defective items. Let $X(I,s)$ be an indicator random variable that is equal $1$ if $A(s,I)$ returns a subset $L\subset I$ of size $\ell$ and $0$ otherwise. Then for every $I$, $\E_s[X(s,I)]\ge 2/3$. Therefore, $\E_s[\E_I[X(s,I)]]=\E_I[\E_s[X(s,I)]]\ge 2/3$, where the distribution in $\E_I$ is the uniform distribution. Thus, there is a seed~$s_0$ such that $\E_I[X(s_0,I)]\ge 2/3$. That is, for at least $2{n\choose d}/3$ sets of defective items $I$, the {\it deterministic} algorithm $A(s_0,I)$ returns $L\subseteq I$ of size $\ell$. Now, similar to the proof of Theorem~\ref{ADKL}, the algorithm $A(s_0,I)$ makes at least
$$\log \frac{\frac{2}{3}{n\choose d}}{{n-\ell\choose d-\ell}}\ge \ell\log(n/d)-1.$$\qed
\end{proof}

In particular,
\begin{theorem} \label{ARUL}
Let $\ell\le d\le n/2$ and $d$ is unknown to the algorithm. Any adaptive randomized algorithm that, with probability at least $2/3$, detects $\ell$ defective items must make at least $\ell\log(n/d)-1$ tests.
\end{theorem}

We now prove the upper bound when $d$ is unknown to the algorithm. This proves result~(7) in Figure~\ref{Table2}.
\begin{theorem} Let $\ell\le d/2$ and $d$ be unknown to the algorithm. There is a polynomial time adaptive randomized algorithm that detects $\ell$ defective items and makes $\ell\log (n/d)+\ell\log\log(1/\delta)+O(\ell+\log \log (\min(n/d,d))+\log(1/\delta))$ tests. 
\end{theorem}
\begin{proof}
We first estimate $d$ within a factor of $2$ and probability at least $1-\delta/2$. By Lemma~\ref{Estd}, this can be done in $2\log\log (n/d)+O(\log(1/\delta))$. Then, by Theorem~\ref{ARKU}, the result follows.\qed
\end{proof}

\section{Non-Adaptive and Deterministic}
In this section, we study the test complexity of non-adaptive deterministic algorithms. 

For the upper bound, we need the following definition.
\begin{definition}
A $(r,s)$-{\it restricted weight one $t\times n$-matrix} $M$ is a $t\times n$ 0-1-matrix such that any $r$ columns in $M$ contains at least $s$ distinct weight one vectors. 

That is, for every $r$ distinct columns $j_1,j_2,\ldots,j_r\in [n]$ in $M$, there are $s$ rows $i_1,i_2,\ldots,i_s\in [t]$ such that $\{(M_{i_k,j_1},M_{i_k,j_2},\ldots,M_{i_k,j_r})\}_{k=1,\ldots,s}$ are $s$ distinct vectors of weight one. 
\end{definition}
The following is obvious.
\begin{lemma}\label{scale} Let $\ell<s$.
If $M$ is $(r,s)$-restricted weight one $t\times n$-matrix, then it is $(r-\ell,s-\ell)$ and $(r,s-\ell)$-restricted weight one $t\times n$-matrix.
\end{lemma}

We prove the following simple properties of such a matrix.
\begin{lemma}\label{a234} Let $n>d>\ell>0$.
Let $M$ be a $(2d,d+\ell)$-restricted weight one $t\times n$-matrix. Let $Q^{(i)}=\{j|M_{i,j}=1\}$ for $i\in [t]$.
\begin{enumerate}
\item\label{a2341} For every two sets $A\subset B\subseteq [n]$ where $|A|=d$ and $|B|=2d$, there is $Q^{(i)}$ such that $Q^{(i)}\cap A=\emptyset$ and $Q^{(i)}\cap B\not=\emptyset$. 
\item For every $C\subseteq E\subseteq [n]$ where $|C|=d$ and $|E|\le 2d$ there are $\ell$ sets $Q^{(i_1)},\ldots, Q^{(i_\ell)}$ such that $|Q^{(i_j)}\cap E|=|Q^{(i_j)}\cap C|=1$ and for every $j_1\not=j_2$, $Q^{(i_{j_1})}\cap C\not=Q^{(i_{j_2})}\cap C$.
\end{enumerate}
\end{lemma}
\begin{proof}
Consider the columns of $M$ with the indices of $A$ and $B$. There are $d+\ell$ distinct weight one vectors in the columns with indices $B$. Since $d+\ell>d$ and $|A|=d$, one of those vectors is zero in all the indices of $A$. Therefore, $M$ contains a row $i$ that is zero in the indices of $A$ and of weight one on the indices of $B$.
Thus, $Q^{(i)}$ satisfies $Q^{(i)}\cap A=\emptyset$ 
and $Q^{(i)}\cap B\not=\emptyset$. This proves~\ref{a2341}.

Assume that $|E|=2d$. Otherwise, add $2d-|E|$ new items to $E$. Consider the columns of $M$ with the indices of $E$ and $C\subseteq E$. There are $d+\ell$ distinct weight one vectors in the columns of $M$ with indices of $E$. Since $C\subset E$ and $|E\backslash C|=d$, at least $\ell$ of those vectors are zero in the indices of $E\backslash C$ and weight one in the indices of $C$. Let $i_1,\ldots, i_\ell$ be the rows that correspond to those vectors. Then $|Q^{(i_j)}\cap E|=|Q^{(i_j)}\cap C|=1$, and for every $j_1\not=j_2$, $Q^{(i_{j_1})}\cap C\not=Q^{(i_{j_2})}\cap C$.\qed
\end{proof}

We now prove
\begin{lemma}\label{constR} Let $s\le cr$ for some constant $1/2<c<1$ and
$$t=O\left(\frac{r\log(n/r)+\log(1/\delta)}{\log(1/c)}\right).$$ consider a $t\times n$ 0-1-matrix $M$ where $M_{i,j}=1$ with probability $1/r$.
Then, with probability at least $1-\delta$, $M$ is a $(r,s)$-restricted weight one $t\times n$-matrix.

In particular, there is a $(r,s)$-restricted weight one $t\times n$-matrix with
$$t=O\left(\frac{r\log(n/r)}{\log(1/c)}\right).$$
\end{lemma}
\begin{proof}
 Consider any $r$ columns $J=\{j_1,\ldots,j_{r}\}$ in $M$. Let $A_J$ be the event that columns $J$ in $M$ do not contain at least $s$ distinct weight one vectors. For every $i\in [t]$, the probability that $(M_{i,j_1},\ldots,M_{i,j_r})$ is of weight $1$ is ${r\choose 1}(1/r)(1-1/r)^{r-1}\ge 1/2$. In every such row, the entry that is equal to $1$ is distributed uniformly at random over $J$.  Let $m_J$ be the number of such rows. The probability that columns $J$ in $M$ do not contain at least $s$ distinct weight one vectors is at most
$$\Pr[A_J|m_J=m]\le {r\choose s-1}\left(\frac{s-1}{r}\right)^{m}\le 2^rc^{m}.$$ 
Since $\E[m_J]\ge t/2$, by Chernoff's bound, 
$$\Pr\left[m_J<\frac{t}{4}\right]\le 2^{-t/16}.$$
Therefore, the probability that $M$ is not $(r,s)$-restricted weight one $t\times n$-matrix is at most
\begin{eqnarray*}
\Pr[(\exists J\subset [n], |J|=r) A_J]&\le& {n\choose r} \Pr[A_J]\\
&\le& {n\choose r} (\Pr[A_J|m_J\ge t/4]+\Pr[m_J<t/4])\\
&\le& {n\choose r}\left(2^rc^{t/4}+2^{-t/16}\right)\\
&\le& {n\choose r}2^{r+1}c^{t/16}\le \delta
\end{eqnarray*}\qed
\end{proof}

We now show how to use the $(r,s)$-restricted weight one matrix for testing.
\begin{lemma}\label{AlgDN} Let $D$ be an integer. If there is a $t\times n$-matrix such that for every $D/4\le d'\le 4D$, $M$ is $(2d',d'+\ell)$-restricted weight one matrix, then there is a non-adaptive deterministic algorithm that, when $d/4\le D\le 4d$ is known in advance to the algorithm, detects $\ell$ defective items and makes $t$ tests.  
\end{lemma}
\begin{proof} Since $d/4\le D\le 4d$, we have $D/4\le d\le 4D$, and therefore, the matrix $M$ is $(2d,d+\ell)$-restricted weight one $t\times n$-matrix. We give the algorithm. The tests of the algorithm are $Q^{(i)}=\{j|M_{i,j}=1\}$, $i\in[t]$. The following is the algorithm:
\begin{enumerate}
    \item Let $Answer_i=T_I(Q^{(i)})$.
    \item Let $X=[n]$; $Y=\emptyset$.
    \item For $i=1$ to $t$
    \item \hspace{.3in} If $Answer_i=0$ then $X\gets X\backslash Q^{(i)}$.
    \item For $i=1$ to $t$
    \item \hspace{.3in} If $Answer_i=1$ and $|Q^{(i)}\cap X|=1$ then $Y\gets Y\cup (Q^{(i)}\cap X)$.
    \item Output $Y$.
\end{enumerate}
Let $X'=X$ after executing steps 3-4. We first show that $|X'|<2d$ and $I\subset X'$, i.e., $X'$ contains all the defective items. First, if $Answer_i=T_I(Q^{(i)})=0$, then $Q^{(i)}\cap I=\emptyset$, and therefore, $I\subset X\backslash Q^{(i)}$. Thus, by step~4, $I\subseteq X'$. 

By step 4, it follows that if $T_I(Q^{(i)})=0$, then $X'\cap Q^{(i)}=\emptyset$.
Now, assume to the contrary that $|X'|\ge 2d$. Consider any subset $X''\subset X'$ of size $|X''|=2d$. Since $|I|= d$, by Lemma~\ref{a234}, there is $Q^{(j)}$ such that $Q^{(j)}\cap I=\emptyset$ and $Q^{(j)}\cap X''\not=\emptyset$. Therefore, $T_I(Q^{(j)})=0$ and $X'\cap Q^{(j)}\not=\emptyset$. A contradiction. 

Now $I\subseteq X'$, $|I|=d$ and $|X'|\le 2d$. By Lemma~\ref{a234}, there are $\ell$ sets $Q^{(i_1)},\ldots,Q^{(i_\ell)}$ such that $|Q^{(i_j)}\cap I|=|Q^{(i_j)}\cap X'|=1$, and for every $j_1\not=j_2$, $Q^{(i_{j_1})}\cap I\not=Q^{(i_{j_2})}\cap I$. Therefore, step~6 detects at least $\ell$ defective items.\qed
\end{proof}

We are now ready to prove the upper bound. This proves~(9) in Figure~\ref{Table2}.
\begin{theorem} Let $\ell\le D/8$. There is a non-adaptive deterministic algorithm that, when $d/4\le D\le 4d$ is known in advance to the algorithm, detects $\ell$ defective items and makes $O(d\log(n/d))$ tests.  
\end{theorem}
\begin{proof}
Since $d/4\le D\le 4d$, we have $D/4\le d\le 4D$. We construct a $\left(r,s\right)$-restricted weight one $t\times n$-matrix where $r=8D$ and $s=7\frac{3}{4}D+\ell$. Since $\ell\le D/8$, we have $s/r\le 0.985$, and by Lemma~\ref{constR}, there is a $\left(8D,7\frac{3}{4}D+\ell\right)$-restricted weight one $t\times n$-matrix with 
$$t=O\left(8D\log \frac{n}{8D}\right)=O\left(d\log \frac{n}{d}\right).$$
Let $D/4\le d'\le 4D$. By 
 Lemma~\ref{scale}, $M$ is also a $(8D,7\frac{3}{4}D+\ell-(d'-D/4)=8D-d'+\ell)$-restricted weight one $t\times n$-matrix and $(8D-(8D-2d'),(8D-d'+\ell)-(8D-2d'))=(2d',d'+\ell)$-restricted weight one $t\times n$-matrix. Then, by Lemma~\ref{AlgDN}, the result follows. \qed
\end{proof}

We now prove the lower bound. This proves~(10) in Figure~\ref{Table2}.
\begin{theorem}\label{NDKL}
Suppose some integer $D$ is known in advance to the algorithm where $d/4\le D\le 4d$. Any non-adaptive deterministic algorithm that detects one defective item must make at least $\Omega(d\log(n/d))$ tests.
\end{theorem}
\begin{proof}
Consider any non-adaptive deterministic algorithm ${\cal A}$ that detects one defective item. Let $M$ be a 0-1-matrix of size $t\times n$ that their rows are the 0-1-vectors that correspond to the tests of ${\cal A}$. That is, if $Q^{(i)}$ is the $i$th test of ${\cal A}$, then the $i$th row of $M$ is $(M_{i,1},\ldots,M_{i,n})$ when $M_{i,j}=1$ if $j\in Q^{(i)}$ and $M_{i,j}=0$ otherwise.
Let $M^{(i)}$ be the $i$th column of $M$. Let $I=\{i_1,i_2,\ldots,i_w\}\subseteq [n]$ be any set of size $w\in \{d,d+1,\ldots,2d\}$. If $I$ is the set of defective items, then $\vee_{i\in I}M^{(i)}$ (bitwise or) is the vector of the answers of the tests of ${\cal A}$. Suppose that when $I$ is the set of defective items, ${\cal A}$ outputs $i_j$. Consider the case when the set of defective items is $I'=I\backslash \{i_j\}$. 
Since the answer of ${\cal A}$ on $I'$ is different from the answer on $I$, and  ${\cal A}$ is deterministic, we must have $\vee_{i\in I}M^{(i)}\not=\vee_{i\in I'}M^{(i)}$. Therefore, columns $I$ must contain a vector of weight one in $M$. So far, we have proved that every $w\in\{d,d+1,\ldots,2d\}$ columns in $M$ contains a vector of weight one. 

This also implies that if $J\subset [n]$ and $|J|\in \{d,d+1,\ldots,2d\}$, then $\oplus_{j\in J}M^{(j)}\not=0$ (bitwise xor). This is because if $\oplus_{j\in J}M^{(j)}=0$, then the columns $J$ do not contain a vector of weight one. 

Now consider the maximum size subset $J_0\subset [n]$, $|J_0|<d$ such that $\oplus_{j\in J_0}M^{(j)}=0$. We claim that there is no set $J'\subset [n]\backslash J_0$, $|J'|\le d$, such that $\oplus_{j\in J'}M^{(j)}=0$. This is because if such $J'$ exists, then $\oplus_{j\in J_0\cup J'}M^{(j)}=0$. Then if $|J_0\cup J'|\in \{d,d+1,\ldots,2d\}$, we get a contradiction, and if $|J_0\cup J'|<d$, then $|J_0\cup J'|>|J_0|$ and $J_0$ is not maximum, and again we get a contradiction. Therefore, no set $J'\subset [n]\backslash J_0$, $|J'|\le d$ satisfies $\oplus_{j\in J'}M^{(j)}=0$. 

Consider the sub-matrix $M'$ composed of the $n-|J_0|\ge n-d$ columns $[n]\backslash J_0$ of $M$. The above property shows that every $2d$ columns in $M'$ are linearly independent over the field $GF(2)$. Then the result immediately follows from the bounds on the number of rows of the parity check matrix in coding theory~\cite{RonRothBook}. We give the proof for completeness. 

We now show that the xor of any $d$ columns in $M'$ is district from the xor of any other $d$ columns. If there are two sets of $d$ columns $J_1$ and $J_2$ that have the same xor, then $\oplus_{j\in J_1\Delta J_2}M^{(j)}=0$ and $|J_1\Delta J_2|\le 2d$. A contradiction. Therefore, by summing all the possible $d$ columns of $M'$, we get ${n\choose d}$ distinct vectors.  
Thus, the number of rows of $M'$ is at least
$$t\ge \log {n\choose d}=\Omega \left(d\log\frac{n}{d}\right).$$\qed
\end{proof}

We now prove the lower bound when $d$ is unknown to the algorithm. This proves result~(11) in Figure~\ref{Table2}. Result (12) follows from the algorithm that tests every item individually. 
\begin{theorem}\label{NDUL}
If $d$ is unknown, then any non-adaptive deterministic algorithm that detects one defective item must make at least $\Omega(n)$ tests.
\end{theorem}
\begin{proof}
Consider any non-adaptive deterministic algorithm ${\cal A}$ that detects one defective item. Let $M$ be a 0-1-matrix of size $t\times n$ whose rows correspond to the tests of ${\cal A}$. 

Suppose for the set of defective items $I_0=[n]$ the algorithm outputs $i_1$, for the set $I_1=[n]\backslash \{i_1\}$ outputs $i_2$, for $I_2=[n]\backslash\{i_1,i_2\}$ outputs $i_3$, etc. Obviously, $\{i_1,\ldots,i_n\}=[n]$. Now, since the output for $I_0$ is distinct from the output for $I_1$, we must have a row in $M$ that is equal to $1$ in entry $i_1$ and zero elsewhere. Since the output for $I_1$ is distinct from the output for $I_2$, we must have a row in $M$ that is equal to $1$ in entry $i_2$ and zero in entries $[n]\backslash \{i_1,i_2\}$. Etc. Therefore, $M$ must have at least $n$ rows.\qed
\end{proof}

\section{Non-Adaptive and Randomized}
In this section, we study the test complexity of non-adaptive randomized algorithms. 

We will use the following for the upper bound.
\begin{lemma} \label{detectOne}
There is a non-adaptive deterministic algorithm that makes $t=\log n+0.5\log\log n+O(1)$ tests and decides whether $d\le 1$ and if $d=1$ detects the defective item. 
\end{lemma}
\begin{proof}
We define a 0-1-matrix $M$, where the rows of the matrix correspond to the tests of the algorithm. The size of the matrix is $t\times n$, where $t$ is the smallest integer such that $n\le {t\choose \lfloor t/2\rfloor}$ and its columns contain distinct Boolean vectors of weight $\lfloor t/2\rfloor$. Therefore $t=\log n+0.5\log\log n+O(1)$.

Now, if there are no defective items, we get $0$ in all the answers of the tests. If there is only one defective item, and it is $i\in [n]$, then the vector of answers to the tests is equal to the column $i$ of $M$. If there is more than one defective item, then the weight of the vector of the answers is greater than $\lfloor t/2\rfloor$. \qed
\end{proof}

For the upper bound, we prove the following. This proves (13) in Figure~\ref{Table2}.
\begin{theorem}\label{NRKU}
Suppose some integer $D$ is known in advance to the algorithm where $d/4\le D\le 4d$. There is a polynomial time non-adaptive randomized algorithm that makes $O(\ell\log(n/d)+\log(1/\delta)\log(n/d))$ tests and, with probability at least $1-\delta$, detects $\ell$ defective items.
\end{theorem}

\begin{proof} If $\ell\ge D/32$, then the non-adaptive randomized algorithm that finds all the defective items makes $O(d\log(n/d))=O(\ell\log(n/d))$ tests. So, we may assume that $\ell<D/32\le d/8$.

Let $\ell\le d/8$. The algorithm runs $t=O(\ell+\log(1/\delta))$ iterations. At each iteration, it uniformly at random chooses each element in $X=[n]$ with probability $1/(2D)$ and puts it in $X'$. If $|X'|>4n/D$, then it continues to the next iteration. If $|X'|\le 4n/D$, then it uses the algorithm in Lemma~\ref{detectOne} to detect if $X'$ contains one defective item, and if it does, it detects the item. If $X'$ contains no defective item or more than one item, it continues to the next iteration.

Although the presentation of the above algorithm is adaptive, it is clear that all the iterations can be run non-adaptively. 

Let $A$ be the event that $X'$ contains exactly one defective item. The probability of $A$ is
$$\Pr[A]={d\choose 1}\frac{1}{2D}\left(1-\frac{1}{2D}\right)^{d-1}\ge \frac{1}{10}.$$
Since $\E[|X'|]=n/(2D)$, by Chernoff's bound
$$\Pr[|X'|>4n/D]\le \left(\frac{e^7}{8^8}\right)^{n/(2D)}\le \frac{1}{20}.$$
Therefore, 
$$\Pr[A\mbox{\ and\ }|X'|\le 4n/D]\ge \frac{1}{20}.$$
Now, assuming $A$ occurs, the defective in $X'$ is distributed uniformly at random over the $d$ defective items. Since $\ell<d/8$, at each iteration, as long as the algorithm does not get $\ell$ defective items, the probability of getting a new defective item in the next iteration is at least $7/8$. Let $B_i$ be the event that, in iteration $i$, the algorithm gets a new defective item.
Then 
$$\Pr[B_i]=\frac{7}{8}\Pr[A\mbox{\ and\ }|X'|\le 4n/D]\ge \frac{7}{160}.$$
By Chernoff's bound, after $O(\ell+\log(1/\delta))$ iterations, 
with probability at least $1-\delta$, the algorithm detects $\ell$ defective items.

Therefore, by Lemma~\ref{detectOne}, the test complexity of the algorithm is
$$O((\ell+\log(1/\delta))\log {|X'|})=O\left(\ell\log\frac{n}{d}+\log(1/\delta)\log\frac{n}{d}\right).$$ \qed
\end{proof}

The following lower bound follows from Theorem~\ref{ARKL}. This proves (14) in Figure~\ref{Table2}.
\begin{theorem} \label{NRKL}
Let $\ell\le d\le n/2$ and $d$ be known in advance to the algorithm. Any non-adaptive randomized algorithm that, with probability at least $2/3$, detects $\ell$ defective items must make at least $\ell\log(n/d)-1$ tests.
\end{theorem}

The following Theorem proves the upper bound for non-adaptive randomized algorithms when $d$ is unknown to the algorithm. This proves result (15) in Figure~\ref{Table2}.
\begin{theorem} Let $c<1$ be any constant, $\ell\le n^c$, and $d$ be unknown to the algorithm. There is a polynomial time non-adaptive randomized algorithm that makes $O(\ell\log^2n+\log(1/\delta)\log^2n)$ tests, and with probability at least $1-\delta$, detects $\ell$ defective items. 
\end{theorem}
\begin{proof}
We make all the tests of the non-adaptive algorithm that, with probability at least $1-\delta/2$, $1/4$-estimate $d$, i.e., finds an integer $D$ such that $d/4<D<4d$. By Lemma~\ref{NonAFd}, this can be done with $O(\log(1/\delta)\log n)$ tests.  

We also make all the tests of the non-adaptive algorithms that, with probability at least $1-\delta/2$, detects $\ell$ defective items for all $d=2^i\ell$, $i=1,2,\ldots,\log (n/\ell)$. By Theorem~\ref{NRKU}, this can be done with
$$O\left(\sum_{i=1}^{\log (n/\ell)}\ell\log\frac{n}{2^i\ell}+\log\frac{2}{\delta}\log\frac{n}{2^i\ell}\right)=O((\ell+\log(1/\delta))\log^2n)$$
tests. \qed
\end{proof}

We now prove the lower bound when $d$ is unknown to the algorithm. This proves result~(16) in Figure~\ref{Table2}. 
\begin{theorem} \label{NRUL}
Let $c<1$ be any constant, $\ell\le n^c$, and $d$ be unknown to the algorithm. Any non-adaptive randomized algorithm that, with probability at least $3/4$, detects $\ell$ defective items must make at least $$\Omega\left(\frac{\ell\log^2 n}{\log\ell+\log\log n}\right)$$ tests.
\end{theorem}
\begin{proof} If $n^c\ge \ell\ge n^{1/32}$, then for $d=n^{(1+c)/2}>\ell$, by Theorem~\ref{NRKL}, the lower bound is
$$\Omega(\ell\log(n/d))=\Omega\left(\frac{\ell\log^2 n}{\log\ell+\log\log n}\right).$$
Therefore, we may assume that $\ell<n^{1/32}$.

Suppose, to the contrary, there is a non-adaptive randomized algorithm ${\cal A}(s,I)$ that, with probability at least $3/4$, detects $\ell$ defective items and makes
$$t=\frac{\ell \log^2 n}{3072(\log \ell+\log\log n)}.$$
Here $s$ is the random seeds, and $I$ is the set of defective items. Define the set of integers 
$N_i=\{k|n/(\ell\log n)^{8i+8}\le k<n/(\ell\log n)^{8i}\}$
for $i=1,2,\ldots,r$ where $$r=\frac{\log n}{32(\log \ell+\log\log n)}.$$ Let $t_i$ be a random variable representing the number of tests $Q$ made by ${\cal A}(s,I)$ where $|Q|\in N_i$. Then $t\ge t_1+t_2+\cdots+t_r$ and
$$\frac{\ell \log^2 n}{3072(\log \ell+\log\log n)}=t=\E[t]\ge \sum_{i=1}^r\E[t_i].$$
Therefore, there is $j\in [r]$ such that $$\E[t_j]\le \frac{\E[t]}{r}=\frac{\ell \log n}{96}.$$ 
By Markov's bound, with probability at least $1-4/96=1-1/24$ we have $t_j<(\ell \log n)/4$.

Let $d=(\ell\log n)^{8j+4}$.
Define the following sets random variables: $M_1$ the set of all tests $Q$ that ${\cal A}(s,I)$ makes where $|Q|<n/(\ell\log n)^{8j+8}$, $M_2$ the set of all tests $Q$ that ${\cal A}(s,I)$ makes where $|Q|\ge n/(\ell\log n)^{8j}$ and $M_3$ the set of tests $Q$ that ${\cal A}(s,I)$ makes where $|Q|\in N_j=\{k|n/(\ell\log n)^{8j+8}\le k<n/(\ell\log n)^{8j}\}$. For a set of defective items $I$, let $A_1(I)$ be the event that all the tests in $M_1$ give answers $0$ and $A_2(I)$ the event that all the tests in $M_2$ give answers $1$. 
Let ${\cal D}$ be the distribution over $I\subset [n]$, $|I|=d$, where the items of $I$ are selected uniformly at random without replacement from $[n]$. Let ${\cal D}'$ be the distribution over $I=\{i_1,\ldots,i_d\}\subset [n]$, where the items of $I$ are selected uniformly at random with replacement from $[n]$. Let $B$ be the event that $I$, chosen according to ${\cal D}'$, has $d$ items. Then, since $\ell<n^{1/32}$,
$$\Pr_{{\cal D}'}[\neg B]=1-\prod_{i=1}^{d-1}\left(1-\frac{i}{n}\right)\le \frac{d(d-1)}{2n}\le \frac{(\ell\log n)^{16j+8}}{2n}\le \frac{(\ell\log n)^{16r+8}}{2n}=\frac{(\ell\log n)^8}{2\sqrt{n}}=o(1).$$
We now have
\begin{eqnarray*}
\Pr_{I\sim {\cal D}}[\neg A_1(I)]&\le & \Pr_{I\sim {\cal D}'}[(\exists Q\in M_1)Q\cap I\not=\emptyset]+\Pr_{I\sim {\cal D}'}[\neg B]\\
&\le& t\Pr_{I\sim {\cal D}'}[Q\cap I\not=\emptyset|Q\in M_1]+o(1)\\
&\le& \frac{\ell \log^2 n}{3072(\log \ell+\log\log n)} \left(1-\left(1-\frac{1}{(\ell\log n)^{8j+8}}\right)^d\right)+o(1)\\
&\le& \frac{\ell \log^2 n}{3072(\log \ell+\log\log n)} \frac{d}{(\ell\log n)^{8j+8}}+o(1)\\
\\
&\le& \frac{1}{3072\log\log n} \frac{1}{\ell^3\log^2 n}+o(1)=o(1).
\end{eqnarray*}
and
\begin{eqnarray*}
\Pr_{I\sim {\cal D}}[\neg A_2(I)]&\le & \Pr_{I\sim {\cal D}'}[(\exists Q\in M_2)Q\cap I=\emptyset]+\Pr_{I\sim {\cal D}'}[\neg B]\\
&\le& t\Pr_{I\sim {\cal D}'}[Q\cap I=\emptyset|Q\in M_2]+o(1)\\
&\le& \frac{\ell \log^2 n}{3072(\log \ell+\log\log n)} \left(1-\frac{1}{(\ell\log n)^{8j}}\right)^d+o(1)\\
&\le& \frac{\ell \log^2 n}{3072(\log \ell+\log\log n)} e^{-\frac{d}{(\ell\log n)^{8j}}}+o(1)\\
\\
&\le& \frac{\ell \log^2 n}{3072(\log \ell+\log\log n)} e^{-\ell^4\log^4 n}+o(1)=o(1).\\
\end{eqnarray*}

We now give a non-adaptive randomized algorithm that for $d=(\ell\log n)^{8j+4}$ makes  $(\ell/4) \log n$ tests and with probability at least $2/3$ detects $\ell$ defective items. 
By Theorem~\ref{NRKL}, and since $\ell<n^{1/32}$ and $d=(\ell\log n)^{8j+4}\le (\ell\log n)^{8r+4}=(\ell \log n)^4n^{1/4}\le n^{1/2}/2$, the test complexity is at least 
$$\ell\log\frac{n}{d}-1\ge \frac{1}{2}\ell \log n,$$ and we get a contradiction.

The algorithm is the following. Choose uniformly at random a permutation $\phi:[n]\to [n]$. Consider the tests of algorithm ${\cal A}$. Let $M_i$, $i=1,2,3$, be the sets defined above. Define $$M_i'=\{(a_{\phi(1)},\ldots,a_{\phi(n)})|(a_1,\ldots,a_n)\in M_i\}.$$
Answer $0$ for all the tests in $M'_1$ and $1$ for all the tests in $M'_2$. If $|M'_3|>(\ell/4)\log n$, then return FAIL. Otherwise, make all the tests in  $M'_3$. Give the above answers of the tests to the algorithm ${\cal A}$ and let $L$ be its output. Output $\phi^{-1}(L)=\{\phi^{-1}(i)|i\in L\}$.  

Since $\phi$ is chosen uniformly at random, the new set of defective items $\phi(I)$ is distributed uniformly at random over all the subsets of $[n]$ of size $d$. The probability that the answers for tests in $M'_1$ and $M_2'$ are wrong 
is $o(1)$. The probability that $|M'_3|>(\ell/4)\log n$ is at most $1/24$. By the promises, the failure probability of ${\cal A}$ is at most $1/4$. Therefore, the probability that this algorithm fails is at most $1/4+1/24+o(1)<1/3$. This completes the proof.\qed
\end{proof}

%
%

\bibliographystyle{splncs04}
\bibliography{GroupTesting}

\ignore{

}

\appendix

\section{Estimating d}\label{A}
In this section, we prove Lemma~\ref{Estd}. 

We first prove.
\begin{lemma}\label{Estd2}
Let $\epsilon<1$ be any positive constant. There is a polynomial time adaptive algorithm that makes $O(\log\log (n/d)+\log(1/\delta))$ expected number of tests and with probability at least $1-\delta$ outputs $D$ such that $(1-\epsilon)d\le D\le (1+\epsilon)d$.
\end{lemma}

We first give an algorithm that makes $O(\log\log (n/d))$ expected number of tests and outputs $D$ that with probability at least $1-\delta$ satisfies
\begin{eqnarray}
   \frac{\delta d^2}{4n\log^2(2/\delta)}\le D\le d. \label{TheD} 
\end{eqnarray}
The algorithm is
\begin{enumerate}
    \item $\lambda=2$.
    \item\label{Ret1} Let each $x\in [n]$ be chosen to be in the test $Q$ with probability $1-2^{-\lambda/n}$.
    \item If $T_I(Q)=0$ then $\lambda\gets \lambda^2$; Return to step \ref{Ret1}.
    \item $D=\delta n/(4\lambda).$
    \item Output $D$.
\end{enumerate}

We now prove
\begin{lemma} We have $$\Pr\left[\frac{\delta d^2}{4n\log^2(2/\delta)}\le D\le d\right]\ge 1-\delta.$$ 
\end{lemma}
\begin{proof}
Let $\lambda_i=2^{2^i}$ and $Q_i$ be a set where each $x\in [n]$ is chosen to be in $Q_i\subseteq [n]$ with probability $1-2^{-\lambda_i/n}$, $i=0,1,\cdots$. Let $i'$ be 
such that $\lambda_{i'}< \delta n/(4d)$ and $\lambda_{i'+1}\ge \delta n/(4d)$. Let $D=\delta n/(4\lambda_j)$ be the output of the algorithm. Then, 
since $\lambda_i\le \lambda_{i+1}/2$, we have $\lambda_{i'-t}< \delta n/(2^{t+2}d)$ and
\begin{eqnarray*}
\Pr[D> d]&=&\Pr[\delta n/(4\lambda_j)> d]=\Pr [\lambda_j< \delta n/(4d)]=\Pr[j\in\{0,1,\ldots,i'\}] \\
&=&\Pr[T_I(Q_0)=1 \vee T_I(Q_1)=1 \vee \cdots\vee T_I(Q_{i'})=1]\le \sum_{i=0}^{i'} \Pr[T_I(Q_i)=1]\\
&=& \sum_{i=0}^{i'} (1-2^{-d\lambda_i/n})\le \sum_{i=0}^{i'} \frac{d\lambda_i}{n}\le \cdots +\frac{\delta}{8}+\frac{\delta}{4}\le \frac{\delta}{2}.
\end{eqnarray*}
Also, since $\lambda_j>a$ implies $\lambda_{j-1}>\sqrt{a}$,
\begin{eqnarray*}
\Pr\left[D<\frac{\delta d^2}{4n\log^2(2/\delta)}\right]&=&\Pr\left[\lambda_j\ge \frac{n^2}{d^2}\log^2\frac{2}{\delta}\right]\\
&=&\Pr\left[T_I(Q_{j-1})=0\ \wedge \ \lambda_j\ge \frac{n^2}{d^2}\log^2\frac{2}{\delta}\right] \\&\le& 2^{-d\lambda_{j-1}/n}\le 2^{-\log(2/\delta)}=\frac{\delta}{2}.
\end{eqnarray*}
This completes the proof.
\qed
\end{proof}

\begin{lemma}\label{yy1}
The expected number of tests of the algorithm is $\log\log (n/d)+O(1)$.
\end{lemma}
\begin{proof}
For $(n/d)^{2}> \lambda_k\ge (n/d)$, the probability that the algorithm makes $k+t+1$ tests is less than $$2^{-d\lambda_{k+t}/n}=2^{-d\lambda_k^{2^t}/n}\le 2^{-(n/d)^{2^t-1}}.$$
Therefore the expected number of tests of the algorithm is at most $k+O(1)$. Since $\lambda_k=2^{2^k}< (n/d)^2$, we have $k=\log\log(n/d)+O(1)$. \qed
\end{proof}

We now give another adaptive algorithm that, given that (\ref{TheD}) holds, it makes $\log\log (n/d)+O(\log\log(1/\delta))$ tests and with probability at least $1-\delta$ outputs $D'$ that satisfies $d\delta/8\le D'\le 8d/\delta$. 

By (\ref{TheD}), we have
$$1\le \frac{d}{D}\le H:=\sqrt{\frac{4\log^2(2/\delta)}{\delta}\frac{n}{D}}$$

Let $\tau=\lceil \log (1+\log H)\rceil$. Then $1\le d/D\le 2^{2^\tau-1}$ and $0\le \log(d/D)\le 2^\tau-1$. 

Consider an algorithm that, given a hidden number $0\le i\le 2^\tau-1$, binary searches for $i$ with queries of the form ``Is $i>m$''. Consider the tree $T(\tau)$ that represents all the possible runs of this algorithm, with nodes labeled with $m$. See, for example, the tree $T(4)$ in Figure~\ref{Tree}.  
\begin{figure}[h]
\centering
{\includegraphics[trim={0 3cm 0 0}, width=0.75\textwidth]{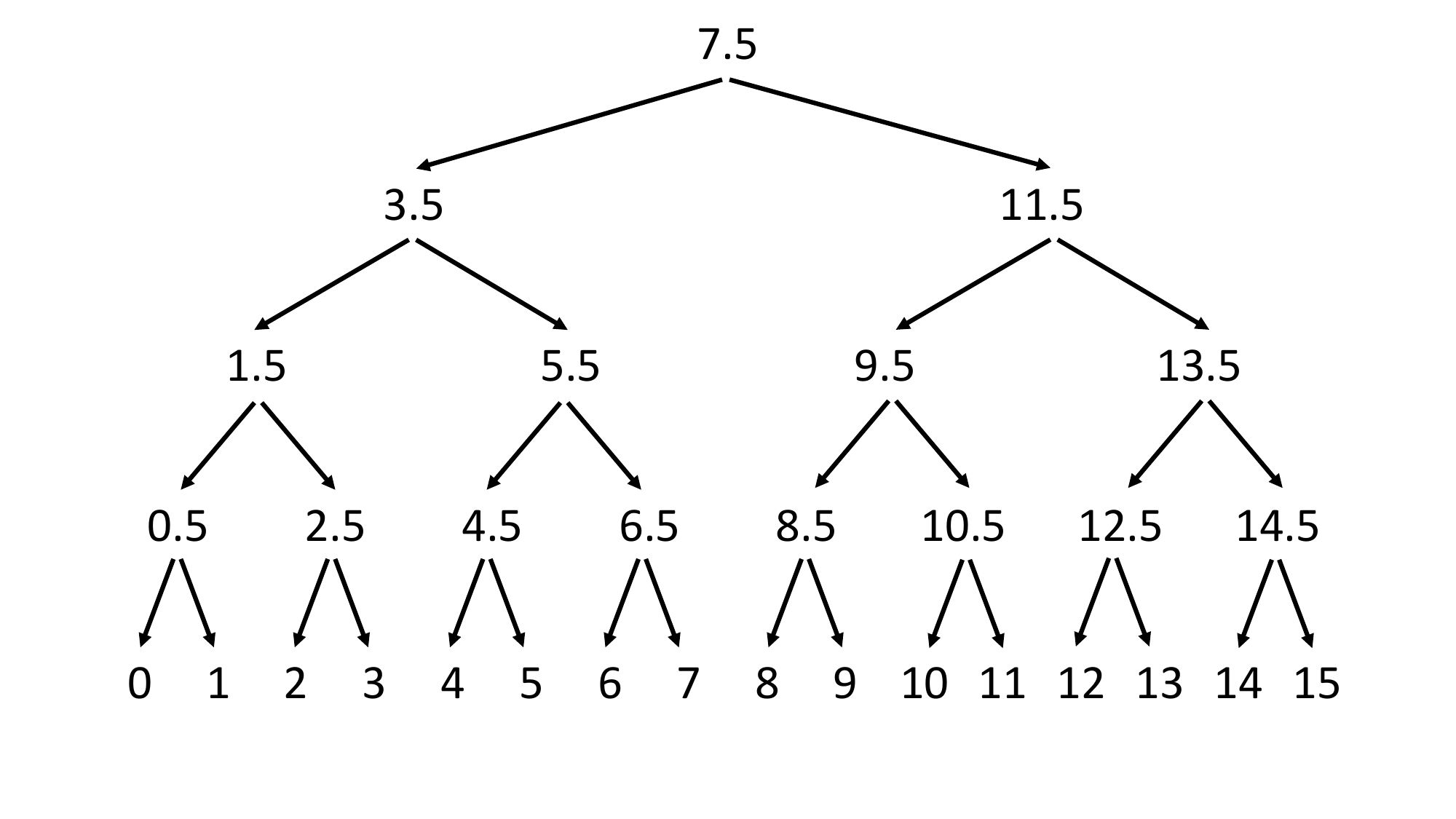}}
\caption{The tree $T(4)$, which is all the runs of the binary search algorithm for $0\le i\le 15$. Suppose we search for the hidden number $i=9$. We start from the tree's root, and the first query is ``Is $i>7.5$''. The answer is yes, and we move to the right son of the root. The following query is ``Is $i>11.5$'' the answer is no, and we move to the left son. Etc.}
\label{Tree}
\end{figure}

We will do a binary search for an integer close to $\log(d/D)$ in the tree $T(\tau)$. 

The algorithm is the following
\begin{enumerate}
    \item Let $\ell=0; r=2^\tau-1;$
    \item While $\ell\not=r$ do
    \item \hspace{.3in} Let $m=(\ell+r)/2$
    \item \hspace{.3in}  Let each $x\in [n]$ be chosen to be in the test $Q$ with probability $1-2^{-1/(2^mD)}$.
    \item \hspace{.3in} If $T_I(Q)=1$ then $\ell= \lceil m\rceil$ else $r=\lfloor m \rfloor$.
    \item Output $D':=D2^\ell$.
\end{enumerate}

We first prove
\begin{lemma}\label{path}
Consider $T(\tau)$ for some integer $\tau$. Consider an integer $0\le i\le 2^\tau-1$ and the path $P_i$ in $T(\tau)$ from the root to the leave $i$. Then 
\begin{enumerate}
    \item $P_i$ passes through a node labeled with $i-1/2$, and the next node in $P_i$ is its right son.
    \item $P_i$ passes through a node labeled with $i+1/2$, and the next node in $P_i$ is its left son.
\end{enumerate}
\end{lemma}
\begin{proof} 
If the path does not go through the node labeled with $i-1/2$ (resp. $i+1/2$), then, in the search, we cannot distinguish between $i$ and $i-1$ (resp. $i+1$). Obviously, if we search for $i$ and reach the node labeled with $i-1/2$, the next node in the binary search is the right son. 
\qed
\end{proof}

Now, by Lemma~\ref{path}, if the algorithm outputs $\ell'$, then there is a node labeled with $m=\ell'-1/2$ that the algorithm went through, and the answer to the test was $1$. That is, the algorithm continues to the right node. 
\begin{eqnarray*}
\Pr\left[D'>\frac{8d}{\delta}\right]&=& 
\Pr\left[D2^\ell>\frac{8d}{\delta}\right]
=\Pr\left[\ell>\log\frac{d}{D}+\log\frac{8}{\delta}\right]\\
&=& \sum_{\ell'=\lceil\log(d/D)+\log(8/\delta)\rceil}^{2^\tau-1} \Pr[\ell=\ell']\\
&=& \sum_{\ell'=\lceil\log(d/D)+\log(8/\delta)\rceil}^{2^\tau-1} \Pr[\mbox{Answer in node labeled with\ }m=\ell'-1/2 \mbox{\ is $1$}]\\
&=&\sum_{\ell'=\lceil\log(d/D)+\log(8/\delta)\rceil}^{2^\tau-1} 1-2^{-d/(2^{\ell'-1/2}D)}\\
\\
&\le&\sum_{\ell'=\lceil\log(d/D)+\log(8/\delta)\rceil}^{2^\tau-1} \frac{d}{D2^{\ell'-1/2}}\le \frac{\delta}{4}+\frac{\delta}{8}\cdots\le \frac{\delta}{2}.\\
\end{eqnarray*}

By Lemma~\ref{path}, if the algorithm outputs $\ell'$, then there is a node labeled with $m=\ell'+1/2$ that the algorithm went through, and the answer to the test was $0$. 
\begin{eqnarray*}
\Pr\left[D'<\frac{\delta d}{8}\right]&=& 
\Pr\left[D2^\ell<\frac{\delta d}{8}\right]
=\Pr\left[\ell<\log\frac{d}{D}-\log\frac{8}{\delta}\right]\\
&=& \sum_{\ell'=0}^{\lfloor\log(d/D)-\log(8/\delta)\rfloor} \Pr[\ell=\ell']\\
&=& \sum_{\ell'=0}^{\lfloor\log(d/D)-\log(8/\delta)\rfloor}\Pr[\mbox{Answer in node labeled with \ }m=\ell'+1/2 \mbox{\ is $0$}]\\
&=&\sum_{\ell'=0}^{\lfloor\log(d/D)-\log(8/\delta)\rfloor} 2^{-d/(D2^{\ell'+1/2})}\\
\\
&\le&2^{-4/\delta}+2^{-8/\delta}+2^{-16/\delta}+\cdots\le \frac{\delta}{4}+\frac{\delta}{8}+\cdots=\frac{\delta}{2}.\\
\end{eqnarray*}
Therefore, with probability at least $1-\delta$, $D'$ satisfies $d\delta/8\le D'\le 8d/\delta$. 

\begin{lemma}\label{yy2}
The number of tests of the algorithm is $\log\log (n/d)+O(\log\log(1/\delta))$.
\end{lemma}
\begin{proof}
Since, by (\ref{TheD}),  ${\delta d^2}/({4n\log^2(2/\delta)})\le D$, the number of tests is 
\begin{eqnarray*}
   \tau+1&\le& \log\log H+3\\
   &\le& 3+\log \log \sqrt{\frac{4\log^2(2/\delta)}{\delta}\frac{n}{D}} \\
   &\le& 3+\log\log \left(\frac{4\log^2\frac{2}{\delta}}{\delta}\cdot\frac{n}{d}\right)=\log\log\frac{n}{d}+O\left(\log\log\frac{1}{\delta}\right).
\end{eqnarray*}
\qed
\end{proof}

Finally, given $D'$ that satisfies  $d\delta/8\le D'\le 8d/\delta$, Falahatgar et al.~\cite{FalahatgarJOPS16} presented an algorithm that, for any constant $\epsilon>0$, makes $O(\log(1/\delta))$ queries and, with probability at least $1-\delta$, returns an integer $D''$ that satisfies $(1-\epsilon)d\le D''\le (1+\epsilon)d$.

By Lemma~\ref{yy1} and~\ref{yy2}, Lemma~\ref{Estd2} follows.

One way to prove Lemma~\ref{Estd} is by running the algorithms in Lemma~\ref{Estd0} and Lemma~\ref{Estd2} in parallel, one step in each algorithm, and halt when one of them halts. Another way is by using the following result.

\begin{lemma}\label{EEE}
Let $d$ and $m$ be integers, and $\epsilon\le 1$ be any real number. There is a non-adaptive randomized algorithm that makes $O((1/\epsilon^2)\log(1/\delta))$ tests and
\begin{itemize}
\item If $d<m$ then, with probability at least $1-\delta$, the algorithm returns $0$.
\item If $d>(1+\epsilon)m$, then, with probability at least $1-\delta$, the algorithm returns $1$.
\item If $m\le d\le (1+\epsilon)m$ then, the algorithm returns $0$ or $1$.
\end{itemize}
\end{lemma}
\begin{proof}
Consider a random test $Q\subseteq X$ where each $x\in X$ is chosen to be in $Q$ with probability $1-(1+\epsilon)^{-1/(m\epsilon)}$. The probability that $T_I(Q)=0$ is $(1+\epsilon)^{-d/(m\epsilon)}$. Since 
\begin{eqnarray*}
\Pr[T_I(Q)=0|d<m]-\Pr[T_I(Q)=0|d>(1+\epsilon)m]&\ge& (1+\epsilon)^{-1/\epsilon}-(1+\epsilon)^{-(1+\epsilon)/\epsilon}\\
&=&(1+\epsilon)^{-1/\epsilon}\frac{\epsilon}{1+\epsilon}\\
&\ge& \frac{\epsilon}{2e}.
\end{eqnarray*}
By Chernoff's bound, we can, with probability at least $1-\delta$, estimate $\Pr[T_I(Q)=0]$ up to an additive error of $\epsilon/(8e)$ using $O((1/\epsilon^2)\log(1/\delta))$ tests. If the estimation is less than $(1+\epsilon)^{-(1+\epsilon)/\epsilon}+\epsilon/(4e)$ we output $0$. Otherwise, we output $1$. This implies the result.\qed
\end{proof}
Now, to prove Lemma~\ref{Estd}, we first run the algorithm in Lemma~\ref{EEE} with $m=\sqrt{n}$ and $\epsilon=1$. If the output is $0$ ($d<2\sqrt n$), then we run the algorithm in Lemma~\ref{Estd0}. Otherwise, we run the algorithm in Lemma~\ref{Estd2}. 

\end{document}